\documentclass[journal]{IEEEtran}

\usepackage{amsmath}
\usepackage{amssymb,mathrsfs} 
\usepackage{graphicx}
\usepackage{amsthm}

\usepackage{subfig}

\usepackage{bbm}
\usepackage{amssymb}  
\usepackage{pifont}
\usepackage{array} 
\usepackage{tabularx} 
\usepackage{epstopdf}
\usepackage{dsfont}
\usepackage{color}
\usepackage{epsfig}
\usepackage{amsthm}
\makeatletter

\newcommand{\Rmnum}[1]{\expandafter\@slowromancap\romannumeral #1@}
\makeatother

\usepackage{balance}
\usepackage[utf8]{inputenc}
\usepackage[english]{babel}

\newtheorem{theorem}{Theorem}
\newtheorem{lemma}{Lemma}
\newtheorem{proposition}{Proposition}
\newtheorem{corollary}{Corollary}

\addto\captionsenglish{\renewcommand{\figurename}{Fig.}}
\usepackage[labelsep=period]{caption}

\long\def\comment#1{}

\newfont{\bbb}{msbm10 scaled 700}

\newfont{\bb}{msbm10 scaled 1100}

\newcommand{\RR}{\mbox{\bb R}}
\newcommand{\PP}{\mbox{\bb P}}

\newcommand{\ZZ}{\mbox{\bb Z}}

\newcommand{\EE}{\mbox{\bb E}}

\newcommand{\xv}{{\bf x}}
\newcommand{\yv}{{\bf y}}

\newcommand{\zerov}{{\bf 0}}

\newcommand{\Am}{{\bf A}}

\newcommand{\Hm}{{\bf H}}
\newcommand{\Id}{{\bf I}}

\newcommand{\Um}{{\bf U}}

\newcommand{\Tt}{\text{T}}

\newcommand{\Hc}{{\cal H}}

\newcommand{\Nc}{{\cal N}}

\newcommand{\Sc}{{\cal S}}

\newcommand{\RNum}[1]{\uppercase\expandafter{\romannumeral #1\relax}}

\newcommand{\Lambdam}{\hbox{\boldmath$\Lambda$}}
\newcommand{\Deltam}{\hbox{\boldmath$\Delta$}}
\newcommand{\Sigmam}{\hbox{\boldmath$\Sigma$}}

\newcommand{\trace}{{\hbox{tr}}}

\newcommand{\eqdef}{\stackrel{\Delta}{=}}

\def\LRT#1#2{\!
\raisebox{.2ex}{$
{{\scriptstyle\;#1}\atop{\displaystyle\gtrless}}
\atop
{\raisebox{-1.25ex}{$\scriptstyle\;#2$}}
$}
\!}

\newcommand{\squeezedequation}{\medmuskip=2mu \thinmuskip=1mu \thickmuskip=3mu}

\begin{document}
%
\title{Stealth Attacks on the Smart Grid}
%
%
%


\author{Ke~Sun,~\IEEEmembership{Student~Member,~IEEE,}
        I\~naki Esnaola,~\IEEEmembership{Member,~IEEE,}
        Samir M. Perlaza,~\IEEEmembership{Senior Member,~IEEE,} \\
       and H. Vincent Poor,~\IEEEmembership{Fellow,~IEEE}
%
\thanks{

K. Sun and I. Esnaola are with the Department of Automatic Control and Systems
Engineering, University of Sheffield, Sheffield S1 3JD, UK.
I. Esnaola is also with the Department of Electrical Engineering, Princeton University, Princeton NJ
08544, USA. (email: ke.sun@sheffield.ac.uk, esnaola@sheffield.ac.uk).

S.~M. Perlaza is with the Institut National de Recherche en Informatique
et Automatique (INRIA), Lyon, France, and also with the Department of
Electrical Engineering, Princeton University, Princeton, NJ 08544 USA (email:
samir.perlaza@inria.fr).

H.~V. Poor is with the Department of Electrical Engineering, Princeton
University, Princeton, NJ 08544 USA (e-mail: poor@princeton.edu).}
}

%



\maketitle

\begin{abstract}

Random attacks that jointly minimize the amount of information acquired by the operator about the state of the grid and the probability of attack detection are presented. The attacks minimize the information acquired by the operator by minimizing the mutual information between the observations and the state variables describing the grid. Simultaneously, the attacker aims to minimize the probability of attack detection by minimizing the Kullback-Leibler (KL) divergence between the distribution when the attack is present and the distribution under normal operation. The resulting cost function is the weighted sum of the mutual information and the KL divergence mentioned above. The trade-off between the probability of attack detection and the reduction of mutual information is governed by the weighting parameter on the KL divergence term in the cost function. The probability of attack detection is evaluated as a function of the weighting parameter. A sufficient condition on the weighting parameter is given for achieving an arbitrarily small probability of attack detection. The attack performance is numerically assessed on the IEEE 14-Bus, 30-Bus, and 118-Bus test systems.

\end{abstract}


%
\IEEEpeerreviewmaketitle

\section{Introduction}
%
%
%
%

\IEEEPARstart{T}{he} smart grid relies on the effective integration of the power grid and advanced communication and sensing infrastructure.
Consistency between the physical layer of the power grid and the energy management system (EMS) in the cyber layer facilitates an economic and reliable operation of the power system.
The 2003 North American outage caused by an alarm system failure \cite{_final_2004} and the 2015 Ukraine power failure caused by the BlackEnergy virus incident \cite{alderson_operational_2016} emphasize the need for cybersecurity mechanisms for the power system.
However, the cybersecurity threats to which the smart grid is exposed are not well understood yet,
and therefore, practical security solutions need to come forth as a multidisciplinary effort combining technologies such as cryptography, machine learning, and information-theoretic security \cite{RS_progress_2016}.

Data injection attacks (DIAs) have emerged as a major source of concern and exemplify the type of cybersecurity threats that specifically target power systems \cite{liu_false_2009}.
DIAs manipulate the state estimation process in the EMS by altering the measurements of the state variables without triggering the bad data detection mechanism put in place by the operator.
In \cite{liu_false_2009}, it is shown that attacks that lie in the column space of the Jacobian measurement matrix are undetectable by testing the residual.
To decrease the number of sensors that need to be compromised by the attacker while remaining undetectable,
the $\ell_0$ norm of the attack vector is used as minimization objective yielding sparse attack in \cite{dan_stealth_2010,sandberg_security_2010,sou_exact_2013} and \cite{kim_strategic_2011}.
%
The case in which sparse attacks are constructed in a distributed setting with multiple attackers is discussed in \cite{tajer_distributed_2011} and \cite{ ozay_sparse_2013}.
Estimation of the operating point is studied in \cite{deng_false_2018} using the power flow or power injection information in an inference problem formulation.
%
Attack detection methods that incorporate the statistical structure of the state variables are presented for the centralized case in \cite{hao_likelihood_2018} and for the decentralized case in \cite{moslemi_fast_2018}.
Similarly, other instances of prior knowledge use include using load forecasts in \cite{ashok_online_2018}, \cite{sridhar_model-based_2014}, and overload mitigation based on corrective dispatch in \cite{che_mitigating_2018}.

The complex nature of the power system leads naturally to a stochastic modelling of the state variables describing the grid.
For instance, the state variables of low voltage distribution systems are well described as following a multivariate Gaussian distribution \cite{genes_recovering_2016}.
DIAs within a Bayesian framework with minimum mean square error estimation  are studied in \cite{kosut_malicious_2011} for the centralized case and in \cite{esnaola_maximum_2016} for the distributed case. However, the fundamental limits governing the performance of attacks in the smart grid are not well understood yet.

%
%
%

Information-theoretic tools are well suited to analyze power system by leveraging the stochastic description of the state variables.
A sensor placement strategy that accounts for the amount of information acquired by the sensing infrastructure is studied in \cite{li_information-theoretic_2013}.
Information-theoretic privacy guarantees for smart meter users are proposed in \cite{varodayan_smart_2011, sankar_smart_2013, tan_increasing_2013} for memoryless stochastic processes and in \cite{arrieta_smart_2017} for general random processes.
In \cite{Sun_information-theoretic_2017}, {stealth} Gaussian DIA constructions are studied in terms of information measures that quantify the information loss and the probability of attack detection induced by the attack.
Therein, the proposed cost function gives the same weight to the information loss and the probability of detection which results in the effective secrecy framework proposed by \cite{hou_effective_2014} in the context of {\it stealth} communications. Stealth DIA constructions are also studied in \cite{dan_stealth_2010,vukovic_network-layer_2011} for the case in which the detection is based on the residual
and in a Bayesian hypothesis testing framework in \cite{huang_bad_2013}.
The approaches in \cite{dan_stealth_2010} and \cite{vukovic_network-layer_2011} consider the minimum cost of compromising the meters and the communication substation, respectively.
On the other hand, \cite{huang_bad_2013} focuses on the delay between the time of attacker launching the attack and the time of operator detecting the attack.

In this paper, the stealth attacks in \cite{Sun_information-theoretic_2017} are generalized by introducing a weight parameter to the objective describing the probability of detection, which allows the attacker to construct attacks with arbitrarily low probability of detection.
Operating under the assumption that the state variables are described by a multivariate Gaussian distribution \cite{kosut_malicious_2011, esnaola_maximum_2016}, we characterize the optimal Gaussian generalized stealth attacks. Since the performance of the attacks depends on the weighting parameter governing the probability of detection, we provide a sufficient condition on the weighting parameter that achieves a desired probability of attack detection.
{To this end, we characterize the probability of attack detection via an upper bound which leverages a concentration inequality in \cite{laurent_adaptive_2000}.}

The rest of the paper is organized as follows: In Section \ref{System_Model}, a Bayesian framework with linearized dynamics for DIA is introduced.
The generalized stealth attack construction and performance analysis are presented in Section \ref{Information-Theoretic_Attack}.
Section \ref{SEC:PDCI} provides the probability of detection of the generalized stealth attack, and the concentration inequality is used to derive the upper bound for probability of detection.
Section \ref{Numerical_Simulation} evaluates the proposed attack performance on IEEE test systems.
%
%
The paper ends with conclusions in Section \ref{Conclusion}.

\section{System Model} \label{System_Model}

\subsection{Bayesian Framework with Linearized Dynamics}

The measurement model for state estimation with linearized dynamics is given by
\begin{IEEEeqnarray}{c}\label{Equ:DCSE}
Y^{m} = \Hm X^{n} + Z^{m},
\end{IEEEeqnarray}
where $Y^{m} \in \RR^{m}$ is a vector of random variables describing the measurements; $X^{n} \in \RR^{n}$ is a vector of random variables describing the state variables;
$\Hm \in \RR^{m\times n}$ is the linearized Jacobian measurement matrix which is determined by the power network topology and the admittances of the branches;
and $Z^{m} \in \RR^{m} $ is the additive white Gaussian noise (AWGN) with distribution $\Nc(\zerov,\sigma^{2} \Id_{m})$ that is introduced by the sensors as a result of the thermal noise, c.f. \cite{abur_power_2004} and  \cite{grainger_power_1994}.
In the remaining of the paper, we assume that the vector of the state variables follows a multivariate Gaussian distribution given by
\begin{IEEEeqnarray}{c}
  X^{n} \sim \Nc(\zerov,\Sigmam_{X\!X}),
\end{IEEEeqnarray}
where $\Sigmam_{X\!X}\in \Sc^{n}_{+}$ is the covariance matrix of the distribution of the state variables and $\Sc^{n}_{+}$ denotes the set of positive semidefinite matrices of size $n\times n$.
As a result of the linearized dynamic in (\ref{Equ:DCSE}), the vector of measurements also follows a multivariate Gaussian distribution denoted by
\begin{IEEEeqnarray}{c}
  Y^{m} \sim \Nc(\zerov,\Sigmam_{Y\!Y}),
\end{IEEEeqnarray}
where $\Sigmam_{Y\!Y} = \Hm\Sigmam_{X\!X}\Hm^{\Tt} + \sigma^{2}\Id_{m} $ is the covariance matrix of the distribution of the vector of measurements.

Data injection attacks corrupt the measurements available to the operator by adding an attack vector to the measurements.
The resulting vector of compromised measurements is given by
\begin{IEEEeqnarray}{c}
\label{eq:measurement_model}
Y^{m}_{A} = \Hm X^{n} + Z^{m} + A^{m},
\end{IEEEeqnarray}
where $A^{m} \in \RR^{m}$ is the attack vector and $Y^{m}_{A} \in \RR^{m } $ is the vector containing the compromised measurements \cite{liu_false_2009}.
Given the stochastic nature of the state variables, it is reasonable for the attacker to pursue a stochastic attack construction strategy.
In the following, an attack vector which is independent of the state variables is constructed under the assumption that the attack vector follows a multivariate Gaussian distribution denoted by
\begin{IEEEeqnarray}{c}
  A^{m} \sim  \Nc (\zerov,\Sigmam_{A\!A}),
\end{IEEEeqnarray}
where $\Sigmam_{A\!A}\in \Sc^{m}_{+}$ is the covariance matrix of the attack distribution.
{The rationale for choosing a Gaussian distribution for the attack vector follows from the fact that for the measurement model in (\ref{eq:measurement_model}) the additive attack distribution that minimizes the mutual information between the vector of state variables and the compromised measurements is Gaussian \cite{SA_TIT_13}.}
Because of the Gaussianity of the attack distribution, the vector of compromised measurements is distributed as
\begin{IEEEeqnarray}{c}\label{Equ:Dis_Ya}
  Y_{A}^{m} \sim \Nc(\zerov,\Sigmam_{Y_{A}\!Y_{A}}),
\end{IEEEeqnarray}
where $\Sigmam_{Y_{A}\!Y_{A}} = \Hm\Sigmam_{X\!X}\Hm^{\Tt} + \sigma^{2}\Id_{m} + \Sigmam_{A\!A} $ is the covariance matrix of the distribution of the compromised measurements.

It is worth noting that the independence of the attack vector with respect to the state variables implies that constructing the attack vector does not require access to the realizations of the state variables.
In fact, knowledge of the second order moments of the state variables and the variance of the AWGN introduced by the measurement process suffices to construct the attack. This assumption significantly reduces the difficulty of the attack construction.

The operator of the power system makes use of the acquired measurements to detect the attack.
The detection problem is cast as a hypothesis testing problem with hypotheses
\begin{IEEEeqnarray}{cl}
\Hc_{0}:  \ & Y^{m} \sim \Nc(\zerov,\Sigmam_{Y\!Y}), \quad \text{versus}  \\
\Hc_{1}:  \ & Y^{m} \sim \Nc(\zerov,\Sigmam_{Y_{A}\!Y_{A}}).
\end{IEEEeqnarray}
The null hypothesis $\Hc_{0}$ describes the case in which the power system is not compromised, while the alternative hypothesis $\Hc_{1}$ describes the case in which the power system is under attack.

Two types of error are considered in hypothesis testing problems,
Type \RNum{1} error
is the probability of accepting $\Hc_{1}$ when $\Hc_{0}$ is the ground truth, i.e. a false alarm or false positive;
and Type \RNum{2} error
is the probability of accepting $\Hc_{0}$ when $\Hc_{1}$ is the ground truth, i.e. a true negative.
The Neyman-Pearson lemma \cite[Proposition \RNum{2}.D.1]{poor_introduction_1994} states that for a fixed probability of Type \RNum{1} error,
the likelihood ratio test (LRT) achieves the minimum Type \RNum{2} error when compared with any other test with an equal or smaller Type \RNum{1} error.
Consequently, the LRT is chosen to decide between $\Hc_{0}$ and $\Hc_{1}$ based on the available measurements.
The LRT between $\mathcal{H}_{0}$ and $\mathcal{H}_{1}$ takes following form:
\begin{equation}\label{LHRT}
L(\yv) \eqdef \frac{f_{Y^{m}_{A}}(\yv)}{f_{Y^{m}}(\yv)} \ \LRT{\Hc_{1}}{\Hc_{0}} \ \tau,
\end{equation}
where $\yv \in \RR^{m}$ is a realization of the vector of random variables modelling the measurements; $f_{Y_A^m}$ and  $f_{Y^m}$ denote the probability density functions (p.d.f.'s) of $Y_A^m$ and  $Y^m$, respectively; and $\tau$ is the decision threshold set by the operator to meet a given false alarm constraint.


\subsection{Information-Theoretic Setting}\label{Subsec:Information_Theoretic_setting}
The mutual information between two random variables is a measure of the amount of information that each random variable contains about the other random variable.
Consequently, the amount of information that the vector of measurements contains about the vector of state variables is determined by the mutual information between the vector of state variables and the vector of measurements.
Information measures have previously been used to quantify the amount of information acquired by different monitoring systems in a smart grid context.
For instance, in \cite{li_information-theoretic_2013} mutual information is used to quantify the amount of information obtained by phasor measurement units from the grid.
Similarly, mutual information is used to quantify the amount of information leaked by smart meters in \cite{varodayan_smart_2011} and \cite{sankar_smart_2013}.

The Kullback-Leibler (KL) divergence between two probability distributions is a measure of the statistical difference between the distributions.
As such, it is a practical measure to quantify the deviation of the measurement statistics with respect to the statistics under normal operating conditions. For instance, in \cite{moslemi_fast_2018} it is used to test abnormal behaviors on the grid.
For the hypothesis testing problem in (\ref{LHRT}), a small value of the KL divergence between $P_{Y^{m}_{A}}$ and $P_{Y^{m}}$ implies that on average the attack is unlikely to be detected by the LRT set by the attacker for a fixed value of $\tau$.

The purpose of the attacker is to disrupt the normal state estimation procedure by
minimizing the information that the operator acquires about the state variables, while guaranteeing that the probability of attack detection is small enough, and therefore, remain concealed in the system.

An information-theoretic framework for the attack construction is adopted in this paper.
To minimize the information that the operator acquires about the state variables from the measurements,
the attacker minimizes the mutual information between the vector of state variables and the vector of compromised measurements.
Specifically, the attacker aims to minimize $I(X^{n};Y_{A}^{m})$.
The rationale for choosing mutual information to measure the evidence acquired by the measurements stems from the fundamental character of information-theoretic measures.
In particular, mutual information describes the amount of information two random variables share, and therefore, it establishes in quantitative terms how much evidence the measurements contain. For that reason, it is natural for the attacker to attempt to minimize the mutual information with the aim of disrupting the monitoring process of the network operator.

On the other hand,
the probability of attack detection is determined by the detection threshold $\tau$ set by the operator and the distribution induced by the attack on the vector of compromised measurements.
%
%
An analytical expression of the probability of attack detection can be described in closed-form as a function of the distributions describing the measurements under both hypotheses. However, the expression is involved in general and it is not straightforward to incorporate it into an analytical formulation of the attack construction. For that reason, we instead consider the asymptotic performance of the LRT to evaluate the detection performance of the operator.
The Chernoff-Stein lemma \cite[Theorem 11.7.3]{cover_elements_2012} characterizes
the asymptotic exponent of the probability of detection when the number of observations of measurement vectors grows to infinity.
%
In our setting, the Chernoff-Stein lemma states that for any LRT and $\epsilon  \in (0,1/2)$, it holds that
\begin{align}
\label{eq:Chernoff-Stein}
\lim_{k \to \infty} \frac{1}{k} \log \beta_{k}^{\epsilon} = -D(P_{Y^{m}_{A}}\|P_{Y^{m}}),
\end{align}
where $D(\cdot \|\cdot)$ is the KL divergence, $\beta_{k}^{\epsilon}$ is the minimum Type II error such that the Type I error $\alpha$ satisfies $\alpha < \epsilon$, and $k$ is the number of $m$-dimensional measurement vectors that are available for the LRT.
Therefore, for the attacker, minimizing the asymptotic detection probability is equivalent to minimizing $D(P_{Y^{m}_{A}}\|P_{Y^{m}})$, where $P_{Y_A^m}$ and  $P_{Y^m}$ denote the probability distributions of $Y_A^m$ and  $Y^m$, respectively.

\section{Information-Theoretic Attack} \label{Information-Theoretic_Attack}
\label{SEC:ITA}
\subsection{Generalized Stealth Attacks}

When these two information-theoretic objectives are considered by the attacker, \cite{Sun_information-theoretic_2017} proposes an stealthy attack construction that combines the two objectives in one cost function, i.e.,
\begin{IEEEeqnarray}{c}\label{Equ:Steallth_Obj}
 I(X^{n};Y^{m}_{A}) \hspace{-0.1em} +  \hspace{-0.2em} D( P_{{Y}^{m}_{A}}\|P_{Y^{m}}) \hspace{-0.2em} =  \hspace{-0.2em} D( P_{X^{n}Y_{A}^{m}}\|P_{X^{n}}P_{Y^{m}}),\hspace{-0.2em} \IEEEeqnarraynumspace
\end{IEEEeqnarray}
where $P_{X^{n}Y_{A}^{m}}$ is the joint distribution of $X^{n}$ and $Y_{A}^{m}$.
The resulting optimization problem to construct the attack is given by
\begin{align} \label{Stealth_Obj}
\underset{A^{m}}{\text{min}} \ D( P_{X^{n}Y_{A}^{m}}\|P_{X^{n}}P_{Y^{m}}).
\end{align}
Therein, it is shown that (\ref{Stealth_Obj}) is a convex optimization problem and the covariance matrix of the optimal Gaussian attack is $\Sigmam_{A\!A} = \Hm\Sigmam_{X\!X}\Hm^{{\text{\rm T}}}$.
%
However, numerical simulations on IEEE test system show that the attack construction proposed above yields large values of probability of detection in practical settings.

To address the issue of high probability of detection, in the following we propose an attack construction strategy that tunes the probability of detection with a parameter that weights the detection term in the cost function.
The resulting optimization problem is given by
\begin{IEEEeqnarray}{c}\label{Equ:WeightSum}
\underset{A^{m}}{\text{min}} \  I(X^{n};Y^{m}_{A})  +  \lambda D( P_{{Y}^{m}_{A}}\|P_{Y^{m}})  \IEEEyesnumber,
\end{IEEEeqnarray}
where $\lambda\geq1$ governs the weight given to each objective in the cost function. It is interesting to note that for the case in which $\lambda=1$ the proposed cost function boils down to the effective secrecy proposed in \cite{hou_effective_2014} and the attack construction in (\ref{Equ:WeightSum}) coincides with that in \cite{Sun_information-theoretic_2017}.
For $\lambda>1$, the attacker adopts a conservative approach and prioritizes remaining undetected over minimizing the amount of information acquired by the operator.
By increasing the value of $\lambda$ the attacker decreases the probability of detection at the expense of increasing the amount of information acquired by the operator via the measurements. The case for $\lambda<1$ requires a different treatment and is left as future work.

\subsection{Optimal Attack Construction}
\label{sec:opt_att}

The attack construction in (\ref{Equ:WeightSum}) is formulated in a general setting. The following propositions particularize the KL divergence and mutual information to our multivariate Gaussian setting.
\begin{proposition} {\rm\cite{cover_elements_2012}}
The KL divergence between $m$-dimensional multivariate Gaussian distributions $ \Nc(\zerov,\Sigmam_{Y_{A}\!Y_{A}}) $ and $\Nc(\zerov,\Sigmam_{Y\!Y}) $ is given by
\begin{IEEEeqnarray}{l}\label{KL}
D( P_{Y^{m}_{A}}\|P_{Y^{m}}) \hspace{-0.1em} = \hspace{-0.1em}  \frac{1}{2} \hspace{-0.1em} \left( \hspace{-0.1em}\log \hspace{-0.1em} \frac{|\Sigmam_{Y\!Y}|}{| \Sigmam_{Y_{A}\!Y_{A}}|} - m +\text{\rm tr}\left(\Sigmam_{Y\!Y}^{-\!1}\Sigmam_{Y_{A}\!Y_{A}} \hspace{-0.1em} \right) \hspace{-0.1em}\right). \IEEEeqnarraynumspace \squeezedequation
\end{IEEEeqnarray}
\end{proposition}
\begin{proposition} {\rm\cite{cover_elements_2012}}
The mutual information between the vectors of random variables $X^{n} \sim \Nc(\zerov,\Sigmam_{X\!X})$ and  $Y_{A}^{m} \sim \Nc(\zerov,\Sigmam_{Y_{A}\!Y_{A}})$ is given by
\begin{align}\label{MI}
 I(X^{n};Y^{m}_{A}) = \frac{1}{2} \log \frac{|\Sigmam_{X\!X}||\Sigmam_{Y_{A}\!Y_{A}}|}{|\Sigmam|},
\end{align}
where $\Sigmam$ is the covariance matrix of the joint distribution of $(X^{n}, Y_{A}^{m})$.
\end{proposition}

Substituting (\ref{KL}) and (\ref{MI}) in (\ref{Equ:WeightSum})  we can now pose the Gaussian attack construction as the following optimization problem:
\begin{IEEEeqnarray}{cl}\label{Equ:Weight_Mod}
\underset{\Sigmam_{A\!A} \in \Sc^{m}_{+}}{\text{min}} \ &-(\lambda - 1)\log|\Sigmam_{Y\!Y} + \Sigmam_{A\!A}| - \log |\Sigmam_{A\!A}+\sigma^{2}\Id_{m}| \IEEEnonumber \\
&\quad + \lambda \trace(\Sigmam_{Y\!Y}^{-\!1}\Sigmam_{A\!A})  \IEEEyesnumber.
\end{IEEEeqnarray}
We now proceed to solve the optimization problem above. First, note that the optimization domain $\Sc^{m}_{+}$ is a convex set. The following proposition characterizes the convexity of the cost function.
\begin{proposition}\label{Stealth_Model}
Let $\lambda \geq 1 $.  Then the cost function in the optimization problem in (\ref{Equ:Weight_Mod}) is convex.
\end{proposition}
\begin{proof}
Note that the term $-\log|\Sigmam_{A\!A}+\sigma^{2}\Id_{m}|$ is a convex function on $\Sigmam_{A\!A} \in \Sc^{m}_{+} $
\cite{boyd_convex_2004}.
Additionally, $-(\lambda - 1)\log|\Sigmam_{Y\!Y} + \Sigmam_{A\!A}|$ is a convex function on $\Sigmam_{A\!A} \in \Sc^{m}_{+} $ when $\lambda \geq 1$. Since the trace operator is a linear operator and the sum of convex functions is convex, it follows that the cost function in (\ref{Equ:Weight_Mod}) is convex on $\Sigmam_{A\!A}\in \Sc^{m}_{+}$.
\end{proof}

\begin{theorem} \label{Stealth_OPT}
Let $\lambda \geq 1$. Then the solution to the optimization problem in (\ref{Equ:Weight_Mod})
 is
 \begin{equation}
 \label{eq:att_cons}
 \Sigmam_{A\!A}^{\star} = \frac{1}{\lambda}\Hm\Sigmam_{X\!X}\Hm^{{\text{\rm T}}}.
 \end{equation}

\end{theorem}

\begin{proof}
Denote the cost function in (\ref{Equ:Weight_Mod})  by $f(\Sigmam_{A\!A})$. Taking the derivative of the cost function with respect to $\Sigmam_{A\!A}$ yields 
\begin{IEEEeqnarray}{ll}
 \frac{\partial  f(\Sigmam_{A\!A})}{\ \partial\Sigmam_{A\!A}} \hspace{-0.2em}
   = &-2(\lambda - 1)(\Sigmam_{Y\!Y} + \Sigmam_{A\!A})^{-\!1} \!-\!  2(\Sigmam_{A\!A}+\sigma^{2}\Id_m)^{-\!1} \IEEEnonumber\\
& \ +  2\lambda \Sigmam_{Y\!Y}^{-\!1} + (\lambda - 1)\text{diag}\left((\Sigmam_{Y\!Y} + \Sigmam_{A\!A})^{-\!1}\right) \IEEEnonumber\\
& \ +  \text{diag}\left((\Sigmam_{A\!A}+\sigma^{2}\Id_m)^{-\!1}\right) -\lambda \text{diag}(\Sigmam_{YY}^{-\!1}). \IEEEeqnarraynumspace
\end{IEEEeqnarray}
Note that the only critical point is $\Sigmam_{A\!A}^{\star} = \frac{1}{\lambda} \Hm\Sigmam_{X\!X}\Hm^{{\text{\rm T}}}$.
Theorem \ref{Stealth_OPT} follows immediately from combining this result with Proposition \ref{Stealth_Model}.
\end{proof}
\begin{corollary}\label{Cor_MI}
The mutual information between the vector of state variables and the vector of compromised measurements induced by the optimal attack construction is given by
\begin{IEEEeqnarray}{ll}
&I(X^n;Y_{A}^n) \IEEEnonumber\\
& \ \ =  \frac 1 2 \log\left | \Hm \mathbf{\Sigma}_{X\!X}\Hm^{\textnormal{\Tt}} \left(\sigma^2\Id_m+\frac{1}{\lambda}\Hm \mathbf{\Sigma}_{X\!X}\Hm^{\textnormal{\Tt}}\right)^{-1}\hspace{-0.7em}+\Id_m\right |. \IEEEeqnarraynumspace
\end{IEEEeqnarray}
\end{corollary}

%
Theorem \ref{Stealth_OPT} shows that the generalized stealth attacks share the same structure of the stealth attacks in \cite{Sun_information-theoretic_2017} up to a scaling factor determined by $\lambda$.
The solution in Theorem \ref{Stealth_OPT} holds for the case in which $\lambda\geq 1$, and therefore, lacks full generality.
However, the case in which $\lambda <1$ yields unreasonably high probability of detection \cite{Sun_information-theoretic_2017} which indicates that the proposed attack construction is indeed of practical interest in a wide range of state estimation settings.
Furthermore the optimization problem in (\ref{Equ:Weight_Mod}) results in a non-convex problem when $\lambda < 1$ and the solution obtained above no longer holds. For this reason the case with $\lambda < 1$ is left as a future research question.

Changing the value of $\lambda$ yields different solutions on the Pareto front of the optimization problem in (\ref{Equ:Weight_Mod}) as we show in the numerical results in Section \ref{Subsec:Numerical_Sensitive}.
For any $\lambda \geq 1$, Theorem \ref{Stealth_OPT} guarantees that the generalized stealth attack is the only Parento efficient solution, i.e. the attack construction that minimizes the mutual information subject to the probability of detection constraint being satisfied. By increasing the value of $\lambda$ the attacker places more importance on the probability of detection than on the mutual information which results in a more conservative attack that disrupts less but is more difficult to detect.

Theorem \ref{Stealth_OPT} also shows that the resulting attack construction is remarkably simple to implement provided that the information about the system is available to the attacker.
Indeed, the attacker only requires access to the linearized Jacobian measurement matrix $\Hm$ and the second order statistics of the state variables, but the variance of the noise introduced by the sensors is not necessary.
To obtain the Jacobian, a malicious attacker needs to know the topology of the grid, the admittances of the branches, and the operation point of the system.
The second order statistics of the state variables on the other hand, can be estimated using historical data.
In \cite{Sun_information-theoretic_2017} it is shown that the attack construction with a sample covariance matrix of the state variables obtained with historical data is asymptotically optimal when the size of the training data grows to infinity.

Corollary \ref{Cor_MI} shows that the mutual information increases monotonically with $\lambda$ and that it asymptotically converges to $I(X^n;Y^m)$, i.e. the case in which there is no attack.
While the evaluation of the mutual information as shown in Corollary \ref{Cor_MI} is straightforward, the computation of the associated probability of detection yields involved expressions that do not provide much insight. For that reason, the probability of detection of optimal attacks is treated in the following section.

\section{Probability of Detection of Generalized Stealth Attacks}
\label{SEC:PDCI}

The asymptotic probability of detection of the generalized stealth attacks characterized in Section \ref{sec:opt_att} is governed by the KL divergence as described in (\ref{eq:Chernoff-Stein}). However in the non-asymptotic case, determining the probability of detection is difficult, and therefore, choosing a value of $\lambda$ that provides the desired probability of detection is a challenging task.
In this section we first provide a closed-form expression of the probability of detection by direct evaluation and show that the expression does not provide any practical insight over the choice of $\lambda$ that achieves the desired detection performance. That being the case, we then provide an upper bound on the probability of detection, which, in turn, provides a lower bound on the value of $\lambda$ that achieves the desired probability of detection.

\subsection{Direct Evaluation of the Probability of Detection}
\label{Sec:DEPD}
Detection based on the LRT with threshold $\tau$ yields a probability of detection given by
\begin{align}\label{Equ:PD_D}
{\sf P}_{\sf D} \eqdef \EE\left[\mathbbm{1}_{\left\{L(Y_A^m) \geq \tau\right\}}\right] ,
\end{align}
where $\mathbbm{1}_{\{\cdot\}}$ is the indicator function. The following proposition particularizes the above expression to the optimal attack construction described in Section \ref{sec:opt_att}.

\begin{lemma}\label{Pro_PD}
The probability of detection of the LRT in (\ref{LHRT}) for the attack construction in (\ref{eq:att_cons}) is given by
\begin{IEEEeqnarray}{l}\label{Pro_PD_1}
{\sf P}_{\sf D}(\lambda)  \eqdef \PP\left[{(U^p)}^{\textnormal{\Tt}} \Deltam U^p \geq \lambda\left(2 \log\tau + \log \left|\Id_{p} + \lambda^{-1}\Deltam\right|\right)\right], \IEEEeqnarraynumspace \squeezedequation
\end{IEEEeqnarray}
where $p\eqdef\text{rank} (\Hm \mathbf{\Sigma}_{X\!X}\Hm^{\textnormal{\Tt}})$, $U^p\in\mathbb{R}^p$ is a vector of random variables with distribution $\Nc(\zerov,  \Id_{p})$, and $\Deltam\in\mathbb{R}^{p\times p}$ is a diagonal matrix with entries given by $(\Deltam)_{i,i}=\lambda_i(\Hm \mathbf{\Sigma}_{X\!X}\Hm^{\textnormal{\Tt}})\lambda_i(\Sigmam_{Y\!Y}^{-\!1})$, where $\lambda_i(\Am)$ with $i=1,\ldots, p$ denotes the $i$-th eigenvalue of  matrix $\Am$ in descending order.
\end{lemma}

\begin{proof}
The probability of detection of the stealth attack is,
\begin{IEEEeqnarray}{ll}
\label{Equ:PD_1}
{\sf P_D}(\lambda) &= \hspace{-0.3em} \int_{\Sc} \mathrm{d}P_{Y_{A}^{m}}\\
&= \hspace{-0.3em}\frac{1}{(2\pi)^{\frac{m}{2}}\left|\Sigmam_{Y_{A}\!Y_{A}}\right|^{\frac12}} \hspace{-0.3em} \int_{\Sc} \hspace{-0.1em} \exp\left \{-\frac{1}{2}\yv^{\Tt}\Sigmam_{Y_{A}\!Y_{A}}^{-\!1}\yv\right \} \mathrm{d} \yv, \label{Equ:PD_2} \IEEEeqnarraynumspace
\end{IEEEeqnarray}
where
\begin{equation}
\Sc =  \{ \yv \in \RR^{m} : L(\yv) \geq \tau\}.
\end{equation}
Algebraic manipulation yields the following equivalent description of the integration domain:
\begin{IEEEeqnarray}{l}
\label{eq:int_dom_2}
\Sc \hspace{-0.2em}  = \hspace{-0.2em}  \left \{ \yv \in \RR^{m}\hspace{-0.4em} : \yv^{\Tt} {\Deltam_0} \yv \hspace{-0.2em} \geq \hspace{-0.2em} 2 \log\tau \hspace{-0.2em}  +  \hspace{-0.2em}  \log |\Id_{m} + \Sigmam_{A\!A}\Sigmam_{Y\!Y}^{-\!1}|\right\},\IEEEeqnarraynumspace
\end{IEEEeqnarray}
with ${\Deltam_0} \eqdef \Sigmam_{Y\!Y}^{-\!1} - \Sigmam_{Y_{A}\!Y_{A}}^{-\!1} $. Let $\Sigmam_{Y\!Y}=\Um_{Y\!Y}\Lambdam_{Y\!Y}\Um_{Y\!Y}^{\Tt}$ where $\Lambdam_{Y\!Y}\in\RR^{m\times m}$ is a diagonal matrix containing the eigenvalues of $\Sigmam_{Y\!Y}$ in descending order and $\Um_{Y\!Y}\in\RR^{m\times m}$ is a unitary matrix whose columns are the eigenvectors of $\Sigmam_{Y\!Y}$ ordered matching the order of the eigenvalues. Applying the change of variable $\yv_{1} \eqdef \Um_{Y\!Y}\yv$ in (\ref{Equ:PD_2}) results in
\begin{IEEEeqnarray}{l}
{\sf P_D}(\lambda)\hspace{-0.2em} =\hspace{-0.2em}
\frac{1}{(2\pi)^{\frac{m}{2}}\hspace{-0.2em}\left|\Sigmam_{Y_{A}\!Y_{A}}\right|^{\frac12}} \!\!\int_{\Sc_1}\hspace{-0.5em}\exp\left \{-\frac{1}{2}\yv^{\Tt}_1\Lambdam_{Y_{A}\!Y_{A}}^{-\!1}\yv_1\right \} \mathrm{d} \yv_1, \label{Equ:PD_3} \IEEEeqnarraynumspace
\end{IEEEeqnarray}
{where $\Lambdam_{Y_{A}\!Y_{A}}\in\RR^{m\times m}$ denotes the diagonal matrix containing the eigenvalues of $\Sigmam_{Y_{A}\!Y_{A}}$ in descending order}.
Noticing that $\Sigmam_{A\!A}$ and $\Sigmam_{Y_{A}\!Y_{A}}$ are also diagonalized by $\Um_{Y\!Y}$,
the integration domain $\Sc_{1}$ is given by
\begin{IEEEeqnarray}{l}
\Sc_{1} \hspace{-0.2em} = \hspace{-0.2em} \left\{ \hspace{-0.1em} \yv_{1} \hspace{-0.2em} \in \hspace{-0.1em} \RR^{m} \hspace{-0.4em} : \yv_{1}^{\Tt} \Deltam_{1} \yv_{1} \hspace{-0.2em} \geq  \hspace{-0.2em} 2 \log\tau \hspace{-0.2em} +  \hspace{-0.2em} \log |\Id_{m} \hspace{-0.2em} + \hspace{-0.2em} \Lambdam_{A\!A}\Lambdam_{Y\!Y}^{-\!1}| \hspace{-0.1em} \right\},\IEEEeqnarraynumspace
\end{IEEEeqnarray}
where $\Deltam_{1} \eqdef \Lambdam_{Y\!Y}^{-\!1}-\Lambdam_{Y_{A}\!Y_{A}}^{-\!1}$ and $\Lambdam_{A\!A}$ denotes the diagonal matrix containing the eigenvalues of $\Sigmam_{{A}\!{A}}$ in descending order. Further applying the change of variable $\yv_{2} \eqdef \Lambdam_{Y_{A}\!Y_{A}}^{-\!\frac{1}{2}} \yv_{1}$ in (\ref{Equ:PD_3}) results in
\begin{equation}
{\sf P_D}(\lambda) =\frac{1}{\sqrt{(2\pi)^{m}}} \int_{\Sc_{2}} \exp\{-\frac{1}{2}\yv_{2}^{\Tt}\yv_{2}\} \mathrm{d} \yv_{2}, \label{Equ:PD_4}
\end{equation}
with the transformed integration domain given by
\begin{IEEEeqnarray}{l}
\Sc_{2} \hspace{-0.2em} = \hspace{-0.2em} \left\{ \yv_{2} \in \RR^{m} \hspace{-0.4em}: \yv_{2}^{\Tt} \Deltam_{2} \yv_{2} \geq 2 \log\tau \hspace{-0.1em} +  \hspace{-0.1em} \log |\Id_{m} \hspace{-0.2em} + \hspace{-0.2em} \Deltam_{2}|\right\},
\IEEEeqnarraynumspace
\end{IEEEeqnarray}
with
\begin{align}
\Deltam_{2} \eqdef \Lambdam_{A\!A} \Lambdam_{Y\!Y}^{-\!1} .
\end{align}
Setting $\Deltam \eqdef \lambda\Deltam_2$ and noticing that $\text{rank}(\Deltam)=\text{rank} (\Hm \mathbf{\Sigma}_{X\!X}\Hm^{\textnormal{\Tt}})$ concludes the proof.
\end{proof}

%
Lemma \ref{Pro_PD} shows that the probability of detection is equivalent to the probability that a weighted sum of independent $\chi^{2}$ random variables exceeds a certain threshold. In our setting, the threshold is determined by the trade-off parameter $\lambda$.
Notice that the left-hand term $(U^{p})^{\textnormal{\Tt}} \Deltam U^{p}$ in (\ref{Pro_PD_1}) is a weighted sum of independent $\chi^{2}$ distributed random variables with one degree of freedom where  the weights are determined by the diagonal entries of $\Deltam$ which depend on the second order statistics of the state variables, the Jacobian measurement matrix, and the variance of the noise; i.e. the attacker has no control over this term.
The right-hand side contains in addition $\lambda$ and $\tau$, and therefore,  the probability of attack detection is described as a function of the parameter $\lambda$.

Unfortunately, no closed-form expression is available for the distribution of a positively weighted sum of independent $\chi^2$ random variables with one degree of freedom \cite{bodenham_comparison_2016}.
Usually, some moment matching approximation approaches, such as the Lindsay–Pilla–Basak (LPB) method \cite{lindsay_moment-based_2000}, are utilized to solve this problem but the resulting expressions are complex and the relation of the probability of detection with $\lambda$ is difficult to describe analytically following this course of action.

In the following an upper bound on the probability of attack detection is derived. The upper bound is then used to provide a simple lower bound on the value $\lambda$ that achieves the desired probability of detection.

\subsection{Upper Bound on the Probability of Detection}

The following theorem provides a sufficient condition for $\lambda$ to achieve a desired probability of attack detection.

\begin{theorem}\label{pro_CI}
Let $\tau > 1$ be the decision threshold of the LRT in (\ref{LHRT}). Given $t>0$ it holds that for all $\lambda \geq \textnormal{max}\left(\lambda^{\star}(t),1\right)$ the probability of attack detection satisfies
\begin{equation}
{\sf P_D}(\lambda)\leq e^{-t},
\end{equation}
where $\lambda^{\star}(t)$ is the only positive solution of $\lambda$ satisfying
\begin{IEEEeqnarray}{c}\label{Equ:CI_SOI}
 2 \lambda \log\tau  - \frac{1}{2\lambda}\textnormal{\trace}(\Deltam^2) -2\sqrt{\textnormal{\trace}({\Deltam}^{2})t} \!-\! 2\|\Deltam\|_{\infty} t = 0,
\IEEEeqnarraynumspace
 \label{Equ:CI_IE}
\end{IEEEeqnarray}
and $\| \cdot \|_{\infty}$ is the infinity norm.
\end{theorem}

\begin{proof}
We start with the result of Lemma \ref{Pro_PD} which gives
\begin{IEEEeqnarray}{l}
\label{Equ:PD_CI1}
{\sf P_{D}}(\lambda) \hspace{-0.2em} = \hspace{-0.2em}\PP \hspace{-0.2em} \left[{(U^p)}^{\textnormal{\Tt}} \hspace{-0.2em} \Deltam U^p \hspace{-0.2em} \geq \hspace{-0.2em} \lambda\left(2 \log\tau \hspace{-0.1em} + \hspace{-0.1em} \log \left|\Id_{p} \hspace{-0.1em}  + \hspace{-0.1em} \lambda^{-1} \hspace{-0.2em} \Deltam\right|\right)\right]. \IEEEeqnarraynumspace
\end{IEEEeqnarray}
We now proceed to expand the term $\log \left|\Id_{p} + \lambda^{-1}\Deltam\right|$ using a Taylor series expansion resulting in
\begin{IEEEeqnarray}{ll}
&\log \left|\Id_{p} + \lambda^{-1}\Deltam\right|  \IEEEnonumber\\
& \ = \sum_{i=1}^{p} \log\left(1+\lambda^{-1}(\Deltam)_{i,i}\right) \\
& \ = \sum_{i=1}^{p}\left(\sum_{j=1}^{\infty}  \left( \frac{\left(\lambda^{-1} (\Deltam)_{i,i}\right)^{2j-1}}{2j-1} - \frac{\left(\lambda^{-1}(\Deltam)_{i,i}\right)^{2j}}{2j}\right) \right).\label{Equ:Taylor_1} \IEEEeqnarraynumspace
\end{IEEEeqnarray}
Since $(\Deltam)_{i,i} \leq 1 \; \textnormal{for} \; i=1, \ldots, p$, and $\lambda \geq 1$,
then
\begin{IEEEeqnarray}{c}
\frac{\left(\lambda^{-1} (\Deltam)_{i,i}\right)^{2j-1}}{2j-1} - \frac{\left(\lambda^{-1}(\Deltam)_{i,i}\right)^{2j}}{2j} \geq 0, \ \textnormal{for}\; j \in \ZZ^+. \IEEEeqnarraynumspace
\end{IEEEeqnarray}
Thus, (\ref{Equ:Taylor_1}) is lower bounded by the second order Taylor expansion, i.e.,
\begin{IEEEeqnarray}{lcl}
\log  \left|\Id_{p} + \Deltam\right| & \ \geq \ & \sum_{i=1}^{p} \left(\lambda^{-1}(\Deltam)_{i,i} -\frac{\left(\lambda^{-1} (\Deltam)_{i,i}\right)^2}{2}\right)\\
& = &
\label{eq:Taylor_2}
\frac{1}{\lambda}\trace(\Deltam) - \frac{1}{2\lambda^2}\trace(\Deltam^2). \IEEEeqnarraynumspace
\end{IEEEeqnarray}
Substituting (\ref{eq:Taylor_2}) in (\ref{Equ:PD_CI1}) yields
\begin{IEEEeqnarray}{l}
 {\sf P_{D}}(\lambda) \hspace{-0.1em} \leq \hspace{-0.1em} \PP \hspace{-0.2em} \left[{(U^p)}^{\textnormal{\Tt}} \hspace{-0.2em} \Deltam U^p \geq  \trace(\Deltam) \hspace{-0.1em} + \hspace{-0.1em} 2  \lambda\log\tau \hspace{-0.1em} - \hspace{-0.1em} \frac{1}{2\lambda}\trace(\Deltam^2)\right] \hspace{-0.2em} .  \label{Equ:PD_CI2} \IEEEeqnarraynumspace
\end{IEEEeqnarray}
Note that $\EE\left[(U^p)^{\textnormal{\Tt}} \Deltam U^p \right]=\trace(\Deltam)$, and therefore, evaluating the probability in (\ref{Equ:PD_CI2}) is equivalent to evaluating the probability of $(U^p)^{\textnormal{\Tt}} \Deltam U^p$ deviating $2 \lambda  \log\tau - \frac{1}{2\lambda}\trace(\Deltam^2)$ from the mean.
In view of this and the results in \cite{laurent_adaptive_2000} and \cite{hsu_tail_2012}, the right-hand side in (\ref{Equ:PD_CI2}) is upper bounded by
\begin{IEEEeqnarray}{ll}
{\sf P_{D}}(\lambda) & \leq \hspace{-0.2em}  \PP \hspace{-0.2em}  \left[{(U^p)}^{\textnormal{\Tt}} \hspace{-0.2em}  \Deltam U^p \geq   \trace(\Deltam) \hspace{-0.2em} + \hspace{-0.2em} 2\sqrt{\trace({\Deltam}^{2})t}  \hspace{-0.2em} + \hspace{-0.2em}  2||\Deltam||_{\infty} t \right] \IEEEeqnarraynumspace \\
& \leq e^{-t},\label{Equ:PD_CI5}
\end{IEEEeqnarray}
for $t >0$ satisfying
\begin{IEEEeqnarray}{c}\label{Equ:CI_SOI}
 2 \lambda\log\tau  - \frac{1}{2\lambda}\trace(\Deltam^2) \geq 2\sqrt{\trace({\Deltam}^{2})t} + 2||\Deltam||_{\infty} t \label{Equ:CI_IE}.
\end{IEEEeqnarray}
The expression in (\ref{Equ:CI_IE}) is satisfied with equality for two values of $\lambda$, one is strictly negative and the other one is strictly positive denoted by $\lambda^{\star}(t)$, when $\tau> 1$. The result follows by noticing that the left-hand term of (\ref{Equ:CI_IE}) increases monotonically for $\lambda>0$ and choosing $\lambda \geq \textnormal{max}\left(\lambda^{\star}(t),1\right)$. This concludes the proof.
\end{proof}
%
It is interesting to note that for large values of $\lambda$ the probability of detection decreases exponentially fast with $\lambda$. We will later show in the numerical results that the regime in which the exponentially fast decrease kicks in does not align with the saturation of the mutual information loss induced by the attack.

\section{Numerical Simulation} \label{Numerical_Simulation}
In this section, we present simulations to evaluate the performance of the proposed attack strategy in practical state estimation settings. In particular,
the IEEE 14-Bus, 30-Bus and 118-Bus test systems are considered in the simulation.
In state estimation with linearized dynamics, the Jacobian measurement matrix is determined by the operation point.
We assume a DC state estimation scenario \cite{abur_power_2004, grainger_power_1994}, and thus, we set the resistances of the branches to $0$ and the bus voltage magnitude to $1.0$ per unit. Note that in this setting it is sufficient to specify the network topology, the branch reactances, real power flow, and the power injection values to fully characterize the system. Specifically, we use the IEEE test system framework provided by MATPOWER \cite{zimmerman_matpower:_2011}.
We choose the bus voltage angle to be the state variables, and use the power injection and the power flows in both directions as the measurements.

As stated in Section \ref{Sec:DEPD}, there is no closed-form expression for the distribution of a positively weighted sum of independent $\chi^2$ random variables, which is required to calculate the probability of detection of the generalized stealth attacks as shown in Lemma \ref{Pro_PD}. For that reason, we use the LPB method and the {MOMENTCHI2} package \cite{bodenham_momentchi2:_2016} to numerically evaluate the probability of attack detection.

The simulation settings are the same as in \cite{Sun_information-theoretic_2017}.
The covariance matrix of the state variables is assumed to be a Toeplitz matrix with exponential decay parameter $\rho$, where the exponential decay parameter $\rho$ determines the correlation strength between different entries of the state variable vector.
The performance of the generalized stealth attack is a function of weight given to the detection term in the attack construction cost function, i.e. $\lambda$, the correlation strength between state variables, i.e. $\rho$, and the Signal-to-Noise Ratio (SNR) of the power system which is defined as
\begin{equation}
\textnormal{SNR} \eqdef 10\log_{10}\left(\frac{\trace{(\Hm\Sigmam_{X\!X}\Hm^\textnormal{T}})}{m\sigma^2}\right).
\end{equation}

\subsection{Generalized Stealth Attack Performance} \label{Subsec:Numerical_GSA}

Fig. \ref{Fig:MIPD_rho_S10} and Fig. \ref{Fig:MIPD_rho_S20} depict the performance of the optimal attack construction given in (\ref{eq:att_cons}) for different values of $\rho$ with $\textnormal{SNR} = 10 \ \textnormal{dB}$ and $\textnormal{SNR} = 20 \ \textnormal{dB}$, respectively, when $\lambda=2$
 and $\tau=2$. Interestingly, the performance of the attack construction does not change monotonically with the correlation strength, which suggests that the correlation among the state variables does not necessarily provide an advantage to the attacker. Admittedly, for a small and moderate values of $\rho$, the performance of the attack does not change significantly with $\rho$ for both objectives. This effect is more noticeable in the high SNR scenario. However, for large values of $\rho$ the performance of the attack improves significantly in terms of both mutual information and probability of detection. Moreover, the advantage provided by large values of $\rho$ is more significant for the 118-Bus system than for the 30-Bus system, which indicates that correlation between the state variables is easier to exploit for the attacker in large systems.
\begin{figure}[t!]
\centering
\includegraphics[scale=0.5]{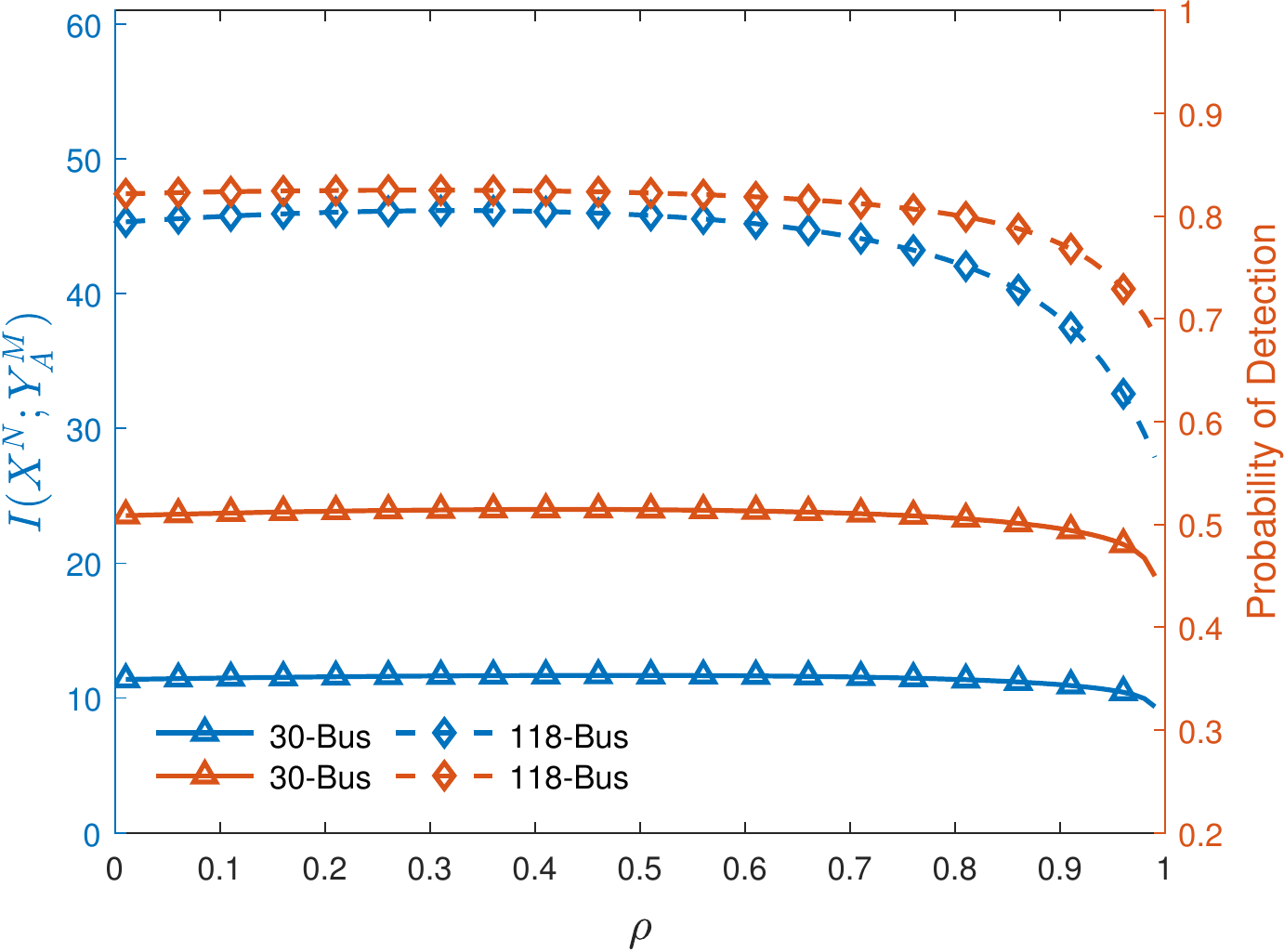}
\caption{ Performance of the generalized stealth attack in terms of mutual information and probability of detection for different values of $\rho$ when $\lambda = 2$, $\tau = 2$, and $\textnormal{SNR} = 10 \ \textnormal{dB}$.}
\label{Fig:MIPD_rho_S10}
\end{figure}

\begin{figure}[t!]
\centering
\includegraphics[scale=0.5]{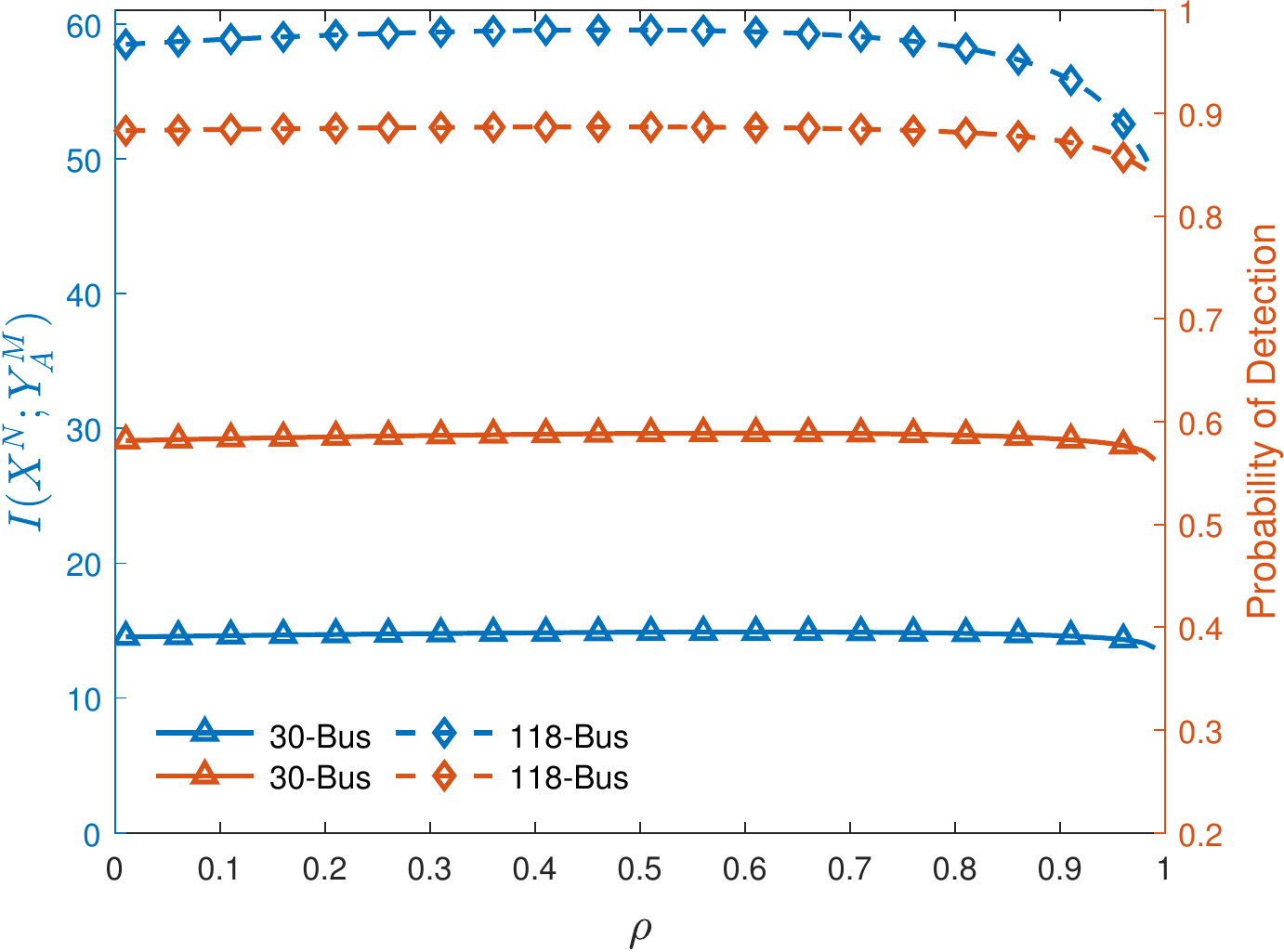}
\caption{ Performance of the generalized stealth attack in terms of mutual information and probability of detection for different values of $\rho$ when $\lambda = 2$, $\tau = 2$, and $\textnormal{SNR} = 20 \ \textnormal{dB}$.}
\label{Fig:MIPD_rho_S20}
\end{figure}

Fig. \ref{Fig:MIPD_lambda_log_S10} and Fig. \ref{Fig:MIPD_lambda_log_S20} depict the performance of the optimal attack construction for different values of $\lambda$ and $\rho$ with $\textnormal{SNR} = 10 \ \textnormal{dB}$ and $\textnormal{SNR} = 20 \ \textnormal{dB}$, respectively, when $\tau=2$.
As expected, larger values of the parameter $\lambda$ yield smaller values of the probability of attack detection while increasing the mutual information between the state variables vector and the compromised measurement vector.
We observe that the probability of detection decreases approximately linearly with respect to $\log \lambda$ for moderate values of $\lambda$.
On the other hand, Theorem \ref{pro_CI} states that for large values of $\lambda$ the probability of detection decreases exponentially fast to zero.
However, for the range of values of $\lambda$ in which the decrease of probability of detection is approximately linear with respect to $\log \lambda$, there is no significant reduction on the rate of growth of mutual information.
In view of this, the attacker needs to choose the value of $\lambda$ carefully as the convergence of the mutual information to the asymptote $I(X^N;Y^M)$ is slower than that of the probability of detection to zero.

\begin{figure}[t!]
\centering
\includegraphics[scale=0.5]{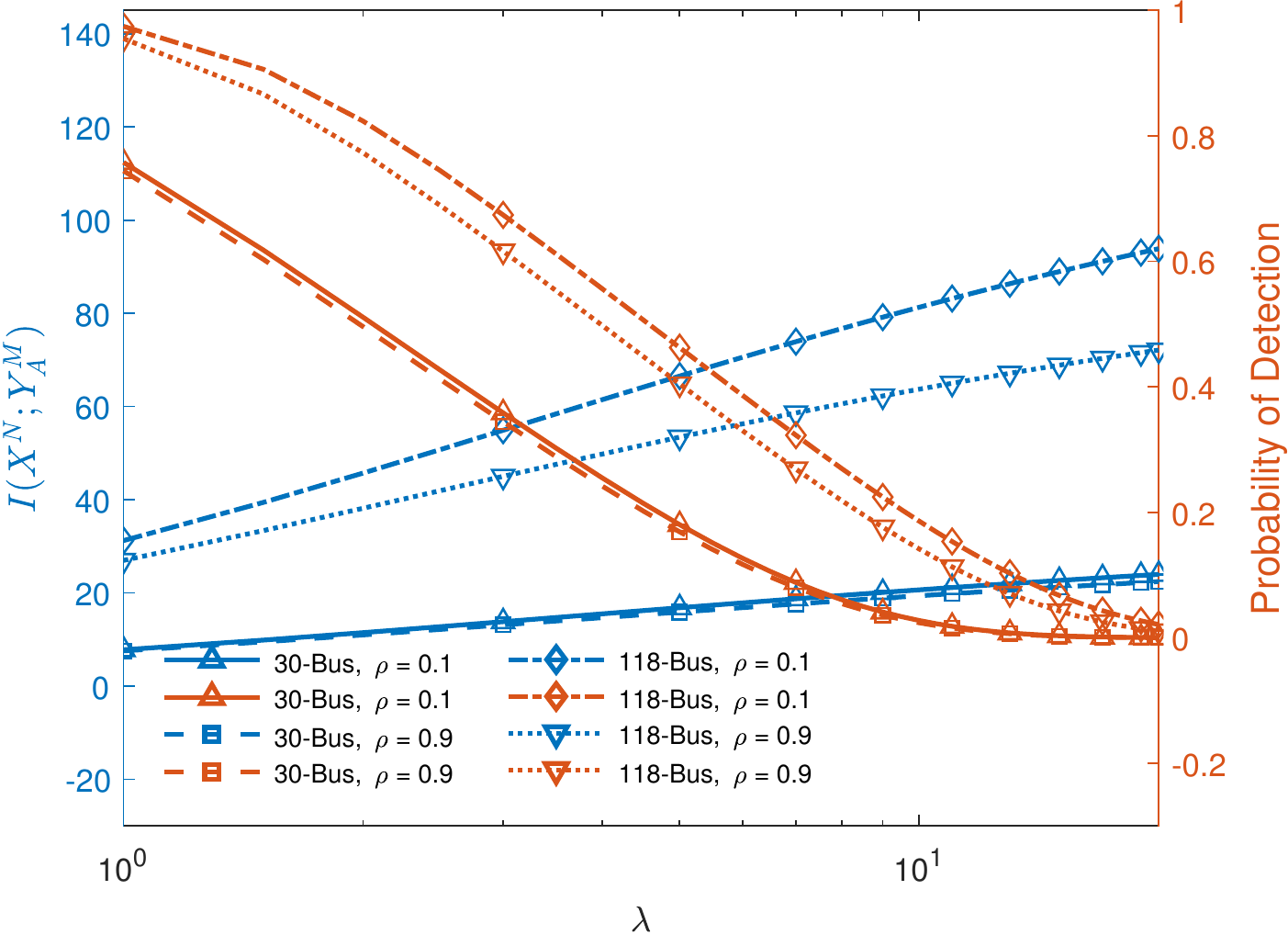}
\caption{ Performance of the generalized stealth attack in terms of mutual information  and probability of detection for different values of $\lambda$ and system size when $\rho = 0.1$, $\rho = 0.9$, $\textnormal{SNR} = 10 \ \textnormal{dB}$ and $\tau = 2$.}
\label{Fig:MIPD_lambda_log_S10}
\end{figure}

\begin{figure}[t!]
\centering
\includegraphics[scale=0.5]{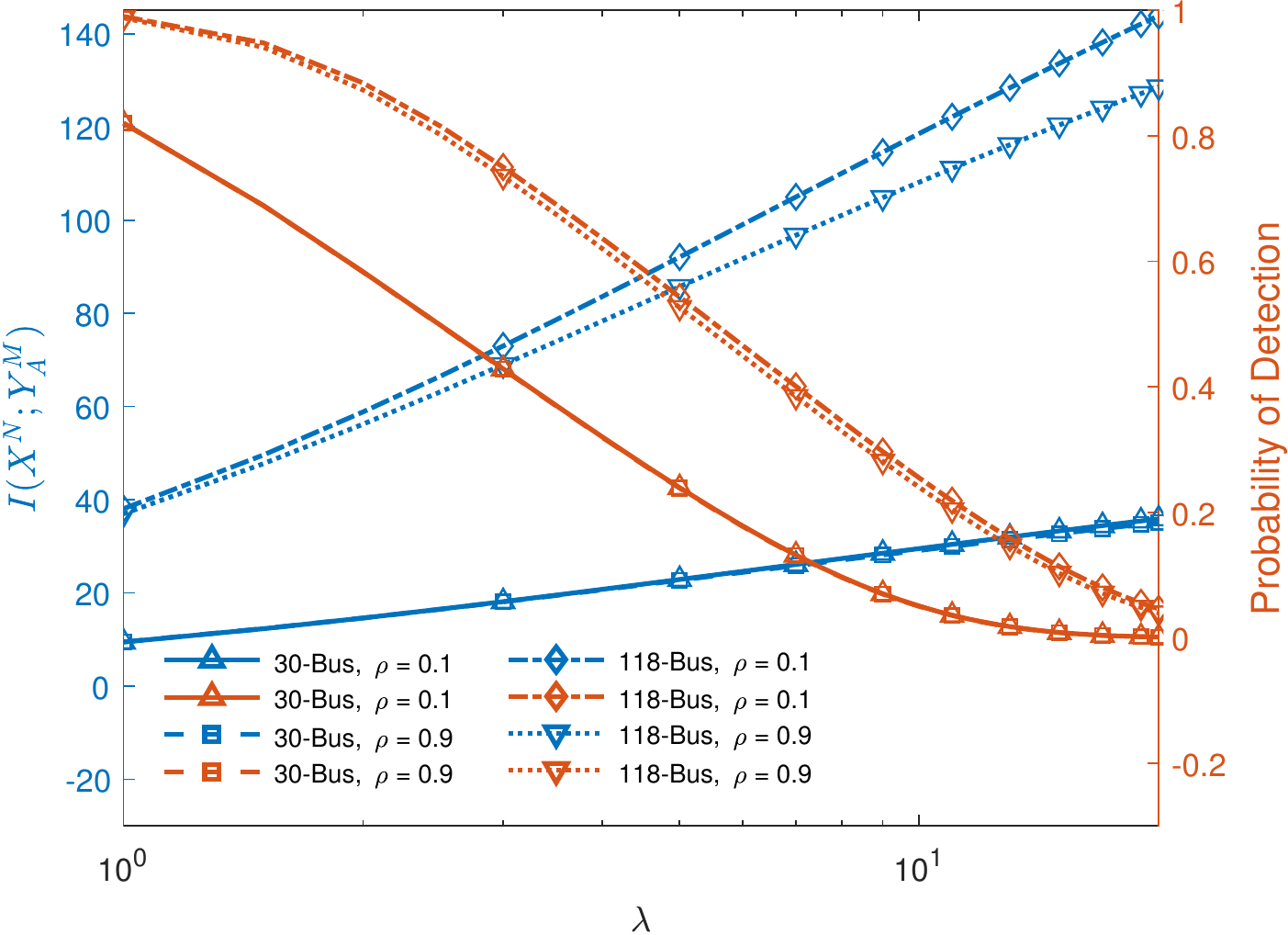}
\caption{ Performance of the generalized stealth attack in terms of mutual information  and probability of detection for different values of $\lambda$ and system size when $\rho = 0.1$, $\rho = 0.9$, $\textnormal{SNR} = 20 \ \textnormal{dB}$ and $\tau = 2$.}
\label{Fig:MIPD_lambda_log_S20}
\end{figure}

The comparison between the 30-Bus and 118-Bus systems shows that for the smaller size system the probability of detection decreases faster to zero while the rate of growth of mutual information is smaller than that on the larger system.
This suggests that the choice of $\lambda$ is particularly critical in large size systems as smaller size systems exhibit a more robust attack performance for different values of $\lambda$.
The effect of the correlation between the state variables is significantly more noticeable for the 118-bus system.
While there is a performance gain for the 30-bus system in terms of both mutual information and probability of detection {due to the high correlation between the state variables}, the improvement is more noteworthy for the 118-bus case.
Remarkably, the difference in terms of mutual information between the case in which $\rho=0.1$ and $\rho=0.9$ increases as $\lambda$ increases which indicates that the cost in terms of mutual information of reducing the probability of detection is large in the small values of correlation.

\begin{figure}[t!]
\centering
\includegraphics[scale=0.4]{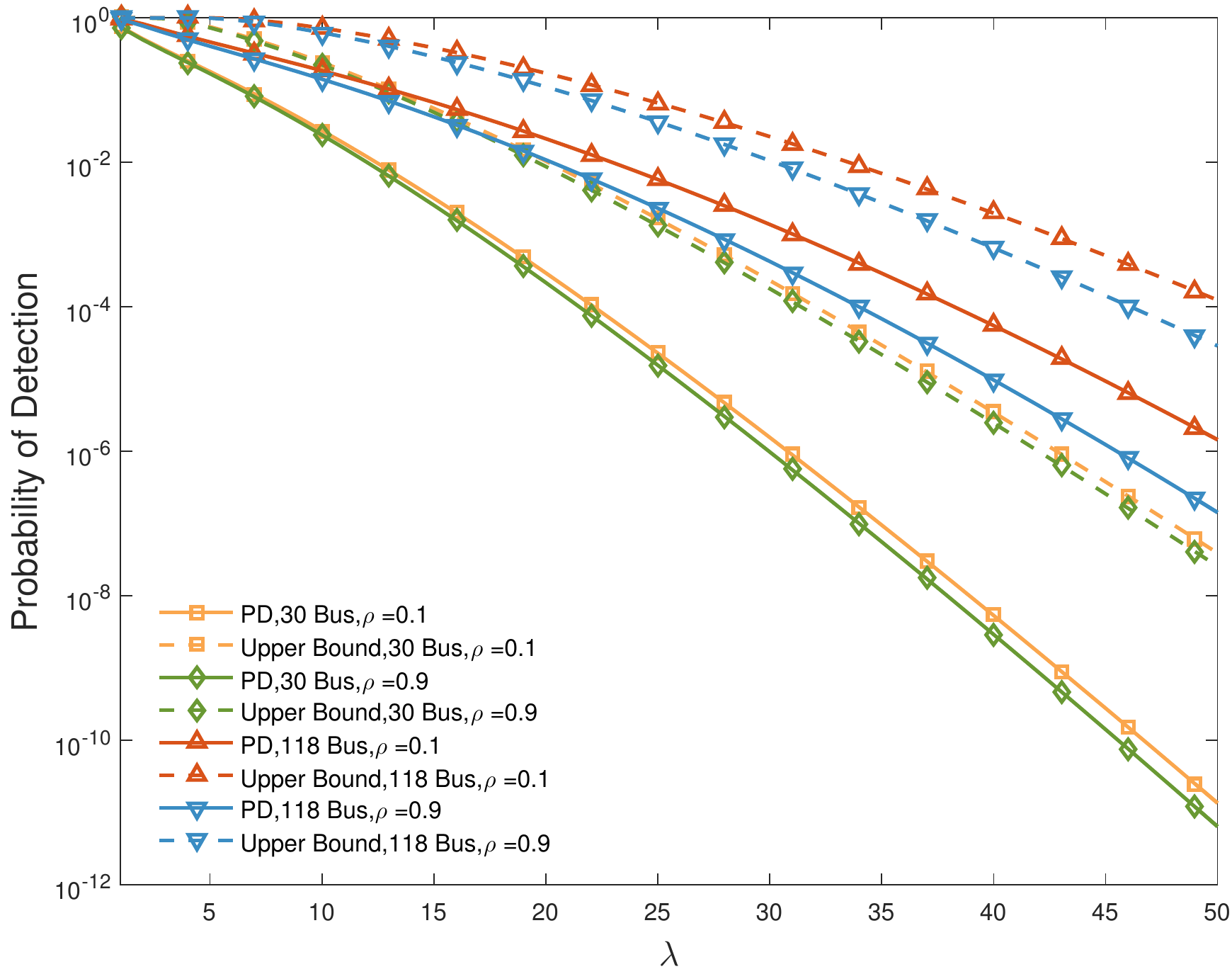}
\caption{Upper bound on probability of detection given in Theorem \ref{pro_CI} for different values of $\lambda$ when $\rho = 0.1 \ \text{or} \ 0.9$, $\textnormal{SNR} = 10 \ \textnormal{dB}$, and $\tau = 2$.}
\label{Fig:CI_Bound_S10}
\end{figure}

\begin{figure}[t!]
\centering
\includegraphics[scale=0.4]{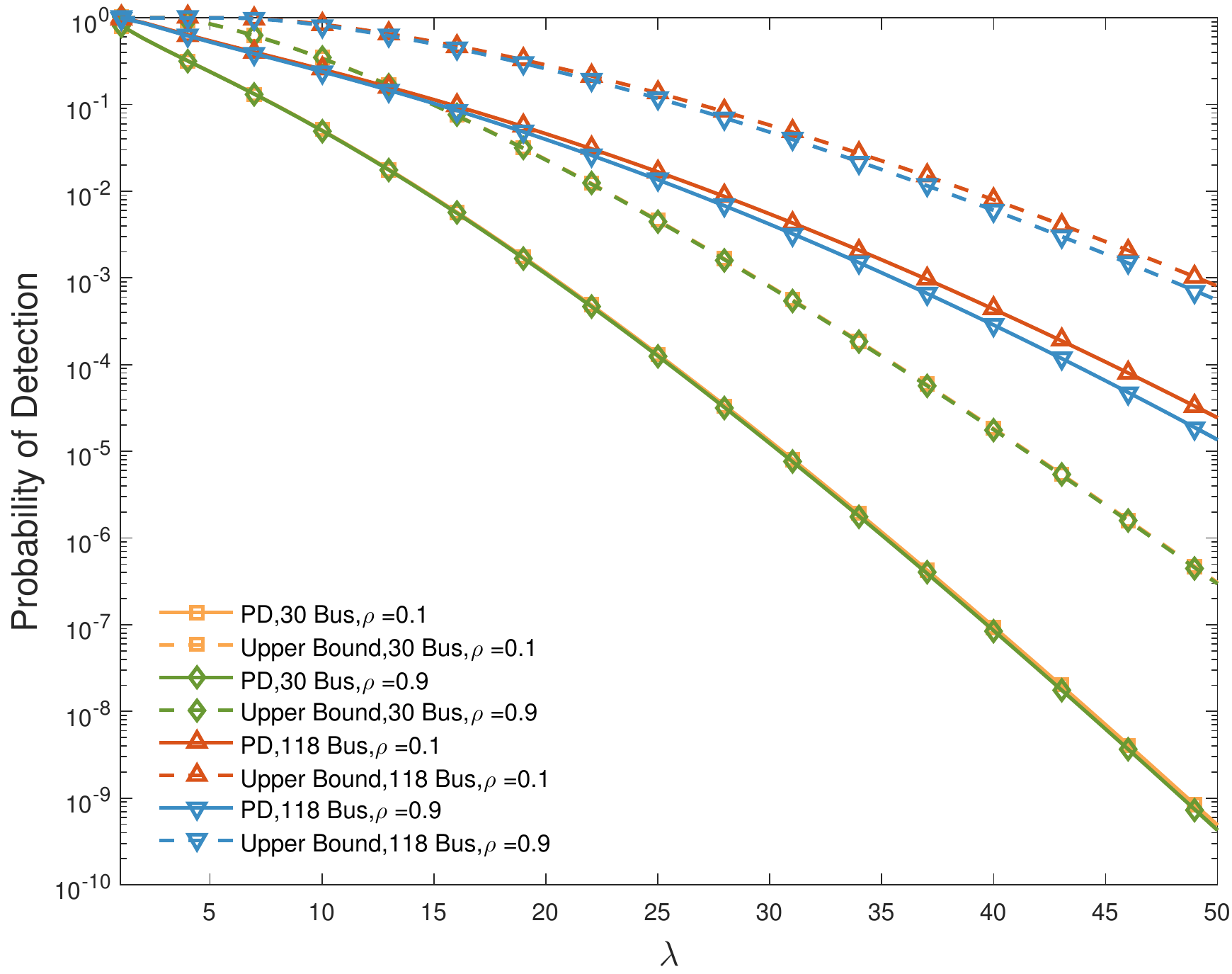}
\caption{Upper bound on probability of detection given in Theorem \ref{pro_CI} for different values of $\lambda$ when $\rho = 0.1 \ \text{or} \ 0.9$, $\textnormal{SNR} = 20 \ \textnormal{dB}$, and $\tau = 2$.}
\label{Fig:CI_Bound_S20}
\end{figure}

The performance of the upper bound given by Theorem \ref{pro_CI} on the probability of detection for different values of $\lambda$ and $\rho$ when $\tau=2$ and $\textnormal{SNR} = 10 \ \textnormal{dB}$ is shown in Fig. \ref{Fig:CI_Bound_S10}. Similarly, Fig. \ref{Fig:CI_Bound_S20} depicts the upper bound with the same parameters but with $\textnormal{SNR} = 20 \ \textnormal{dB}$. As shown by Theorem \ref{pro_CI} the bound decreases exponentially fast for large values of $\lambda$. Still, there is a significant gap to the probability of attack detection evaluated numerically. This is partially due to the fact that our bound is based on the concentration inequality in \cite{laurent_adaptive_2000} which introduces a gap of more than an order of magnitude.  Interestingly, the gap decreases when the value of $\rho$ increases although the change is not  significant. More importantly, the bound is tighter for lower values of SNR for both 30-bus and 118-bus systems.

\subsection{Performance and Sensitivity under AC State Estimation} \label{Subsec:Numerical_Sensitive}

\begin{figure}[t!]
\centering
\includegraphics[scale=0.4]{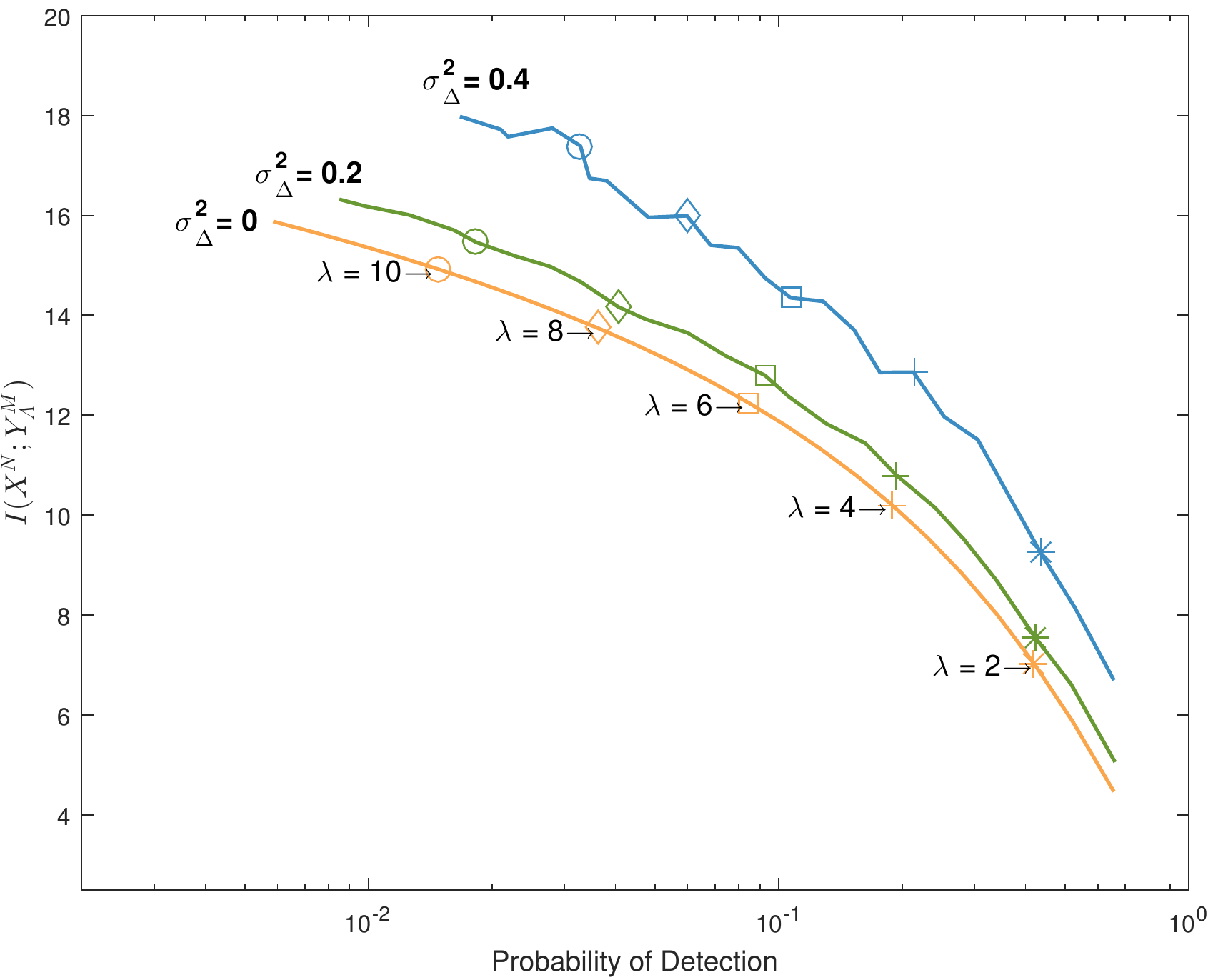}
\caption{Performance of generalized stealth attack in terms of mutual information and probability of detection for different values of $\sigma_{\Delta}^2$ and $\lambda$ on IEEE $14$-Bus system when $\rho = 0.1$, $\tau = 2$, and $\textnormal{SNR} = 20 \ \textnormal{dB}$. The marker represents the same value of $\lambda$ is used in the attack construction.}
\label{Fig:Sen_14}
\end{figure}

\begin{figure}[t!]
\centering
\includegraphics[scale=0.4]{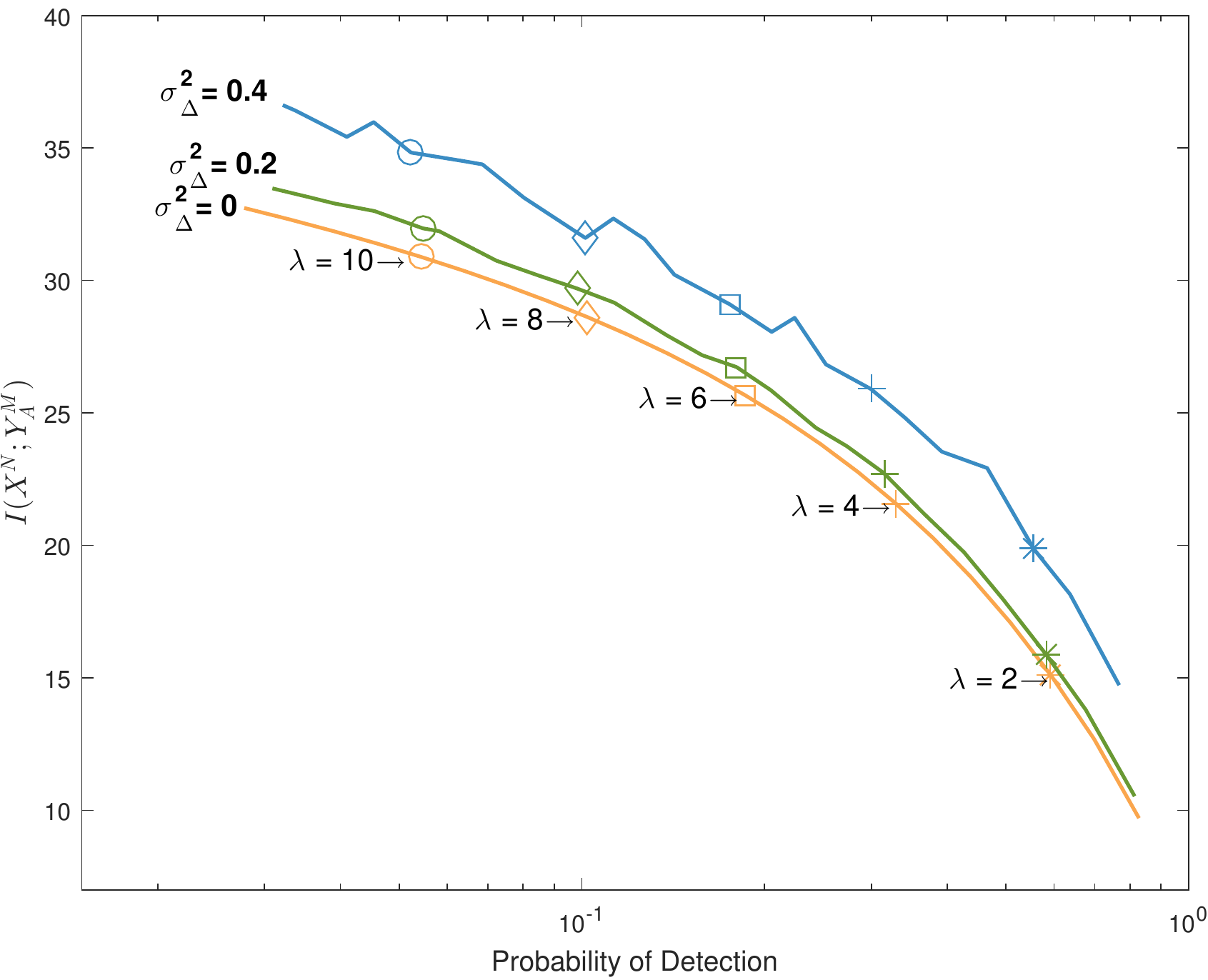}
\caption{Performance of generalized stealth attack in terms of mutual information and probability of detection for different values of $\sigma_{\Delta}^2$ and $\lambda$ on IEEE $30$-Bus system when $\rho = 0.1$, $\tau = 2$, and $\textnormal{SNR} = 20 \ \textnormal{dB}$. The marker represents the same value of $\lambda$ is used in the attack construction.}
\label{Fig:Sen_30}
\end{figure}

In the AC state estimation case the iterative estimation methods require a nominal operation point that is updated for each iteration. When the attacker has the perfect information about the operation point in each iteration, i.e. perfect information about Jacobian matrix $\Hm$ in each iteration, the resulting mutual information and probability of detection follow from Corollary \ref{Cor_MI} and Lemma \ref{Pro_PD} directly.
\emph{In the following, we study the impact of imperfect nominal operation point information on the attack performance.}
In particular the generalized stealth attacks are constructed as $A_{0}^{m} \sim \Nc (\zerov,\frac{1}{\lambda}\Hm_{0}\Sigmam_{X\!X}\Hm_{0}^{\Tt})$, where $\Hm_{0}$ is the Jacobian matrix at the nominal operation point $\xv_{0}$ and is given by
\begin{IEEEeqnarray}{c}
\Hm_{0} = \frac{\partial}{ \partial (X^{n}) } H(X^{n})|_{X^{n} = \xv_{0}},
\end{IEEEeqnarray}
with $ H(X^{n}) \in \RR^{m}$ denoting the vector of random variables induced by the nonlinear relation between the state variables and the measurements.
To model the imperfect knowledge of the nominal point, the nominal linearization point is perturbed with random variable $\Delta \sim \Nc(\zerov,\sigma_\Delta^2\Id)$ resulting in the Jacobian matrix $\Hm$ given by
\begin{IEEEeqnarray}{c}
\Hm = \frac{\partial}{ \partial (X^{n}) } H(X^{n})|_{X^{n} = \xv_{0} + \Delta}.
\end{IEEEeqnarray}
Note that the introduction of this random perturbation gives us a way to control the strength of the perturbation, i.e. the uncertainty over the nominal linearization point, and as a result we study the sensitivity of the attacks under AC state estimation by changing the variance $\sigma^2_\Delta$ in the simulations.

Fig. \ref{Fig:Sen_14} depicts the performance of the generalized stealth attacks in terms of the mutual information and the probability of detection for different values of $\sigma_{\Delta}^2$ and $\lambda$ on the IEEE 14-Bus system when $\rho = 0.1$, $\tau = 2$, and $\textnormal{SNR} = 20 \ \textnormal{dB}$.
Similarly Fig. \ref{Fig:Sen_30} shows the performance of the attacks under the same setting on IEEE 30-Bus system.
We generate $200$ realizations of $\Delta$ per point and for each realization of $\Delta$ we evaluate $2000$ realizations of the state variables.
The curve corresponding to the case when $\sigma_{\Delta}^2=0$ describes the performance of the attacks with perfect knowledge of the nominal operation point.
As expected,
when there is less accurate knowledge about the nominal operation point, i.e. $\sigma_\Delta^2$ increases, the performance of the attack $A_{0}^{m}$ decreases.
Interestingly the performance decrease translates in a larger value of mutual information for all cases.
However, the change in probability of detection is not as significant, to the extent that in some cases the probability of detection decreases.
Note that for all cases, overall the attack performance decreases when perfect operation point is not available.
Interestingly, the stealth of the attacks is more robust for the IEEE 30-Bus system than for the IEEE 14-Bus system, which suggests that the attacker is better positioned to cope with system uncertainty for larger networks.

\section{Conclusions} \label{Conclusion}

We have proposed a novel data injection attacks based on information-theoretic performance measures. Specifically, we have posed the attack construction problem as an optimization problem in which the cost function combines the mutual information and the probability of attack detection. The proposed cost function allows to obtain an arbitrarily small probability of attack detection via a parameter that weights the effect of the mutual information and the probability of detection. The resulting random attack construction has been analyzed in terms of the information loss and the probability of attack detection that it induces on the system. We have characterized the probability of attack detection by obtaining an easy to compute upper bound. The upper bound has been used to provide a practical attack construction guideline by determining the cost function that achieves a given probability of attack detection.

\bibliographystyle{IEEEbib}
\bibliography{reference}

%




\end{document}